\newcommand{\R}{\mathbb R}
\newcommand{\I}{\mathbb I}

\newcommand{\uby}{\eth\Omega}

\documentclass[twoside,a4paper,12pt]{amsart}

\usepackage{times}
\usepackage{dsfont}
\usepackage{epsfig}
\usepackage{amsmath}
\usepackage{amsfonts}
\usepackage{amssymb}
\usepackage{subfigure}
\usepackage[hypertex]{hyperref}

\numberwithin{equation}{section}

\title[Gravitational collapse of homogeneous perfect fluid in HOG theories]{%
Gravitational collapse of homogeneous perfect fluids in higher-order gravity theories}

\author[R.\ Giamb\`o]{Roberto Giamb\`o}
\address{Dipartimento di Matematica e Informatica \hfill\break\indent
Universit\`a di Camerino\hfill\break\indent Italy}
\email{roberto.giambo@unicam.it}

\thanks{This paper is published despite the effects of the Italian law 133/08 (more on \href{http://groups.google.it/group/scienceaction}{http://groups.google.it/group/scienceaction}). This
law drastically reduces public funds to public Italian universities, which is particularly dangerous for free scientific research, and it will
prevent young researchers from getting a position, either temporary or tenured, in Italy. The author is protesting against this law to
obtain its cancelation.}

\date\today

\begin{document}


\theoremstyle{plain}\newtheorem*{teon}{Theorem}
\theoremstyle{definition}\newtheorem*{defin*}{Definition}
\theoremstyle{plain}\newtheorem{teo}{Theorem}[section]
\theoremstyle{plain}\newtheorem{theorem}{Theorem}[section]
\theoremstyle{plain}\newtheorem{prop}[teo]{Proposition}
\theoremstyle{plain}\newtheorem{lemma}[teo]{Lemma}
\theoremstyle{plain}\newtheorem*{lem-n}{Lemma}
\theoremstyle{plain}\newtheorem{cor}[teo]{Corollary}
\theoremstyle{definition}\newtheorem{definition}[teo]{Definition}
\theoremstyle{remark}\newtheorem{remark}[teo]{Remark}
\theoremstyle{plain} \newtheorem{assumption}[teo]{Assumption}
\theoremstyle{definition}\newtheorem{example}{Example} \swapnumbers
\theoremstyle{plain} \newtheorem*{acknowledgement}{Acknowledgements}
\theoremstyle{definition}\newtheorem*{notation}{Notation}


\begin{abstract}
This paper investigates the evolution of collapsing FRW models with a scalar field
having the potential which arises in the conformal frame of high order gravity theories, coupled to matter described by a perfect fluid with energy
density $\rho$ and pressure $p$, obeying a barotropic equation of state. The solutions are shown to evolve generically to a singular state in a finite time and they are used as sources for
radiating objects undergoing complete gravitational collapse. Although these singularities may be naked in some special case, it is shown that generically a black hole forms.
\end{abstract}

\maketitle


\section{Introduction}\label{sec:intro}

The study of Friedmann--Robertson--Walker (FRW) solutions has always been a central topic in relativistic cosmology. In this sense, a great deal of attention is paid to investigation of global structures and, waiting for a satisfying theory of quantum gravity, spacetime boundaries may be described in classical terms only, studying existence of singularities when gravity is coupled to different kind of matter tensors. Especially dealing with early universe, a huge literature exists on scalar field cosmologies, both in the free case \cite{ch}, and in the more complicated situation when a potential is added (see for instance \cite{cole,clw,guetal,ruba}) and initial (i.e. big--bang) singularity existence is investigated.

Scalar field solutions are of great importance also in relativistic astrophysics -- massless solution is a useful model in gravitational collapse studies because the evolution equation in absence of gravity is free of singular solution. On the other side, adding a potential to the model may result again in singularity formation, which here represents the endstate of collapse \cite{joshi}. In this context,
the gravitational collapse  of self--interacting homogeneous scalar field models has been recently analyzed in \cite{collapse}, where it is proven that for a wide class of potential the evolution is generically (i.e. up to a zero--measured initial data set) divergent in a finite comoving time.

In this paper flat FRW spacetime is considered, where
scalar field is nonminimally coupled to
ordinary matter described by a barotropic fluid.  Motivations come from higher order gravity theories (HOG) derived from
Lagrangians of the form
\begin{equation}
L=f\left(  R\right)  \sqrt{-g}+2L_{\mathrm{m}}\left(  \Psi\right)
,\label{hogl}%
\end{equation}
where $f$ is an arbitrary smooth function and $L_{\mathrm{m}}\left(
\Psi\right)  $ is the matter Lagrangian depending on the matter fields $\Psi$.
Nonlinear theories of gravitation are a well-established field of research since the pioneering works \cite{bc,mff,witt}; see also \cite{cafa} for a thorough analysis of interpretation issues.
It is well known that under the conformal transformation%
\begin{equation}
\widetilde{g}_{\mu\nu}=f^{\prime}\left(  R\right)  g_{\mu\nu},\label{conf}%
\end{equation}
the field equations reduce to the Einstein field equations with a scalar field
as an additional matter source, namely
\begin{equation}
\widetilde{G}_{\mu\nu}=T_{\mu\nu}\left(  \widetilde{g},\phi\right)
+\widetilde{T}_{\mu\nu}\left(  \widetilde{g},\Psi\right)  ,\label{confm}%
\end{equation}
where%
\[
T_{\mu\nu}\left(  \widetilde{g},\phi\right)  =\partial_{\mu}\phi\partial_{\nu
}\phi-\frac{1}{2}\widetilde{g}_{\mu\nu}\left[  \left(  \partial\phi\right)
^{2}-2V\left(  \phi\right)  \right]  ,
\]
and
\begin{equation}
\phi=\sqrt{\frac{3}{2}}\ln f^{\prime}\left(  R\right), \label{scfi}%
\end{equation}
allowing to write the potential of the scalar field as
\[
V\left(  R\left(  \phi\right)  \right)  =\frac{1}{2\left(  f^{\prime}\right)
^{2}}\left(  Rf^{\prime}-f\right).
\]

From now on let us work exclusively in the conformal frame and for simplicity drop
the tilde from all quantities. Ordinary matter is described
by a perfect fluid with equation of state $p=(\gamma-1)\rho$ (with $0\leq
\gamma\leq2$) and field equations reduce to the following system:
\begin{subequations}
\begin{align}
&H^{2}+\frac{k}{a^{2}}=\frac{1}{3}\left(  \rho+\frac{1}{2}\dot{\phi}^{2}+V\left(  \phi\right)  \right)  , \label{fri1jm}\\
&\dot{H}=-\frac{1}{2}\dot{\phi}^{2}-\frac{\gamma}{2}\rho+\frac{k}{a^{2}},\label{fri2jm}\\
&\ddot{\phi}+3H\dot{\phi}+V^{\prime}(\phi)=\frac{4-3\gamma}{\sqrt{6}}\rho,\label{emsjm}\\
&\dot{\rho}+3\gamma\rho H=-\frac{4-3\gamma}{\sqrt{6}}\rho\dot{\phi}
\label{conssfjm}
\end{align}
\end{subequations}
(see for example \cite{bbpst,maso}). Here $H=\dot{a}/a$, where $a\left(  t\right)  $ is the scale factor of FRW and
overdot denotes differentiation with respect to time $t$. It is worth stressing here once for all that
this approach is different from the one used in previous investigations \cite{comi,jm} where
the vacuum Lagrangian is conformally transformed into the Einstein frame and the matter Lagrangian is added (see \cite{fara} for a comparative discussion about the different approaches).

Of course, for both cosmological and astrophysical purposes, it is important to study the solutions of \eqref{fri1jm}--\eqref{conssfjm} relaxing as much as possible hypotheses on the potential, although exponential potentials, for instance, have been extensively treated in literature. In the present paper, the class of potentials studied is described in technical terms by Definition \ref{def:V0}. Polynomial potentials with even leading terms, exponential potentials -- with small logarithmic derivative value -- and potentials arising in conformally related theories of gravity, such as \cite{maed}
\begin{equation}
V\left(  \phi\right)  =V_0\left(  1-e^{-\sqrt{2/3}\phi}\right)
^{2} \label{rsquared},
\end{equation}
belong to the class considered. In Section \ref{sec:collapse}, casting the problem into a dynamical systems' framework, it is proven that the qualitative behavior of
the solution is similar to the case without matter, i.e. solutions completely collapse in  finite time up to a zero--measured initial data set (Theorem \ref{thm:main}). Responsible for the formation of this singularity is the scalar field, which dominates over the energy density of the fluid even when the latter diverges.

In Section \ref{sec:star} the solutions are used as source to build collapsing object models where the exterior is given by a generalized Vaidya solution. The endstate of these models is studied for such objects, finding that the collapse generically ends into a black hole. We stress the fact that naked singularities may arise from these models, choosing initial data in the non generical set excluded by Theorem \ref{thm:main}. In this sense, the result is qualitatively similar to what shown in \cite{collapse} -- for the class of potentials considered, weak cosmic censorship hypothesis is not affected when a barotropic perfect fluid is added to the model.

\section{Collapsing solutions in the flat case}\label{sec:collapse}

In the following, we are going to study the future late time behavior of solutions of
\eqref{fri1jm}--\eqref{conssfjm}, with $k=0$, such that they are collapsing at initial
time of observation, i.e. $H(0)<0$. Since $\gamma$ is non negative,
\eqref{fri2jm} implies $H(t)<0$, $\forall t>0$. Regardless of the physical derivation of
the potential $V(\phi)$, discussed in the previous section, the results we are going to state will
hold for the following general class of potential functions $V(\phi)$.

\begin{definition}
\label{def:V} We say that $V$ belongs to the set $\mathfrak{C}$ if the
following conditions are satisfied:

\begin{itemize}
\item[(A1)] (Structure of the set of critical points) The critical points of
$V$ are isolated.

\item[(A2)] (Existence of a suitable bounded "sub-level" set) There exists
$a< b\in\R$ such that
\[
\phi\ge b\Longrightarrow V^{\prime}(\phi)>0,\, \phi\le a \Longrightarrow
V^{\prime}(\phi)<0.
\]
Moreover, setting $V_{*} = \max_{[a,b]}V$, it holds
\[
\lim_{\phi\rightarrow-\infty}V(\phi) > V_{*},\,\lim_{\phi\rightarrow+\infty
}V(\phi) > V_{*}.
\]

\item[(A3)] (Growth condition) The function defined in terms of $V$ and its
first derivative $V^{\prime}$ as
\begin{equation}
\label{eq:Y}u(\phi):=\frac{V^{\prime}(\phi)}{\sqrt6V(\phi)},
\end{equation}
satisfies
\begin{align}
&  \exists \lim_{\phi\to\pm\infty} |u(\phi)|< 1,\label{eq:u-infty}\\
&  \exists\lim_{\phi\to\pm\infty} {u^{\prime}(\phi)} \, (=0).
\label{eq:uprime-infty}%
\end{align}

\end{itemize}
\end{definition}

The proof will be carried out translating the original problem into a compact
framework, in order to apply classical results of dynamical systems. This
approach, already used in \cite{fost} -- where gravity was minimally coupled to
a scalar field under the action of a potential only -- seems more promising in
this framework, where a perfect fluid is added, than the one used in
\cite{collapse}, where everything was reduced to the study of a single second
order ODE. The present approach, however, makes necessary a further assumption on $V(\phi)$ allowing the compactification of the problem, that will be performed once Theorem \ref{thm:unbound} below is proved.

\begin{assumption}
\label{as:phi} We assume that for some constant $M>b$, there exist a
$\mathcal{C}^{2}$ map $f(\phi):[M,+\infty\lbrack\rightarrow]0,s_0]$ such that
$f$ is monotone decreasing (and then $\lim_{\phi\rightarrow+\infty}f(\phi
)=0$, $f(M)=s_0$), there exists $\lim_{\phi\rightarrow+\infty}f^{\prime}(\phi)$ ($=0$)
and moreover
\begin{align}
&  \lim_{\phi\rightarrow+\infty}\frac{u^{\prime}(\phi)}{f^{\prime}(\phi)}=0,\\
&  \exists\lim_{\phi\rightarrow+\infty}\frac{f^{\prime\prime}(\phi)}%
{f^{\prime}(\phi)}\in\mathbb{R}.
\end{align}

\end{assumption}

\begin{definition}\label{def:V0}
We say that a potential $V\in\mathfrak C$ belongs to the class $\mathfrak C_0$ if assumption \ref{as:phi} holds for $V(\phi)$ together with a similar assumption for the behavior at negative infinity,
i.e. supposing also the existence of a function (that we will call $f$ again)
$f:]-\infty,N]\to[-s_0,0[$, where $N<a$, with similar properties as in
Assumption \ref{as:phi} opportunely adapted to this case.
\end{definition}

\begin{remark}\label{rem:counterex}
Although the class $\mathfrak C_0$ contains many relevant examples in cosmology such as polynomials with even leading term, exponential functions $e^{\lambda\phi}$ (with $\lambda<\sqrt 6$, to satisfy growth condition (A3)), or even the potential \eqref{rsquared} arising in the conformal frame theory discussed in the introduction, one can exhibit examples showing that $\mathfrak C\setminus\mathfrak C_0$ is non empty: take for instance $V(\phi)=e^{\phi\mathrm{Si}(\phi)+\cos\phi}$ as $\phi\ge b$, which\footnote{The function $\mathrm{Si}(\phi)$ is the sine integral function, defined as $\int_0^\phi\sin\vartheta/\vartheta\,\mathrm d\vartheta$.}  satisfies the growth assumption (A3) but does not admit a function $f(\phi)$ as in Assumption \ref{as:phi}.

Although the main Theorem \ref{thm:main} stated in Subsection \ref{sec:main} will require $V\in\mathfrak C_0$, for some of the preliminary results (for instance, Theorem \ref{thm:unbound}) it will suffice $V\in\mathfrak C$.
\end{remark}

Let us start the study of the problem introducing the unknown functions
\begin{equation}
x=\frac{1}{H},\qquad y=\frac{\dot{\phi}}{H},\qquad z=\frac{\sqrt{\rho}}%
{H},\label{eq:change}%
\end{equation}
and the time variable change $\mathrm{d}\tau=-H\mathrm{d}t$. The system
\eqref{fri1jm}--\eqref{conssfjm} for $k=0$ becomes
\begin{subequations}
\begin{align}
&  \frac{\mathrm{d}\phi}{\mathrm{d}\tau}=-y,\label{eq:dphi}\\
&  \frac{\mathrm{d}x}{\mathrm{d}\tau}=-\frac{1}{2}x(y^{2}+\gamma
z^{2}),\label{eq:dx}\\
&  \frac{\mathrm{d}y}{\mathrm{d}\tau}=V'(\phi) x^2+3y-\frac{1}{2}y^{3}%
-z^{2}(\alpha+\frac{\gamma}{2}y),\label{eq:dy}\\
&  \frac{\mathrm{d}z}{\mathrm{d}\tau}=-\frac{1}{2}z\left[  y^{2}-\alpha
y+\gamma(z^{2}-3)\right]  ,\label{eq:dz}%
\end{align}
where we denote $\alpha:=\tfrac{4-3\gamma}{\sqrt{6}}$. Equation \eqref{fri1jm}
gives the constraint
\end{subequations}
\begin{equation}
V(\phi)x^{2}+\frac{1}{2}y^{2}+z^{2}=3,\label{eq:constr}%
\end{equation}
which is invariant by the flow of the above system. We also observe, from
\eqref{eq:dx}, that the set $\{x=0\}$ is invariant by the flow and so
$\mathrm{sign}(x)$ is also invariant. This fact guarantees that for collapsing
solutions corresponding to curves with $x<0$, $x(\tau)\,$ will remain
non positive for all $\tau>0$. Let
$\mathbb{I}\subseteq\lbrack0,+\infty)$ be the maximal right interval of
definition. The following crucial result holds.

\begin{theorem}
\label{thm:unbound} Let $V\in\mathfrak C$. For almost every solution $\gamma=(\phi,x,y,z)$ of
\eqref{eq:dphi}--\eqref{eq:dz}, such that \eqref{eq:constr} and $x<0$ are
satisfied at $\tau=0$ (and then $\forall\tau\in\mathbb{I}$), it holds
$$\limsup_{t\rightarrow\sup\mathbb{I}}|\phi(t)|=+\infty.$$
\end{theorem}

\begin{proof}
By contradiction, let $\phi(t)$ be bounded. Since $x$ is negative and monotone, it is bounded on $\I$ and so, using \eqref{eq:constr} and the fact that $V(\phi)$ is bounded from below (possibly by a negative value), the curve $\gamma$ is bounded, which means that $\sup\I=+\infty$. Let $\omega(\gamma)$ be the omega--limit set of $\gamma$ which \cite[Section 3.2]{perko} is a nonempty, connected and compact set invariant by the flow of \eqref{eq:dphi}--\eqref{eq:dz}. By continuity also \eqref{eq:constr} is satisfied on $\omega(\gamma)$.

Let $p\in\omega(\gamma)$ and $\gamma_p$ be the solution through $p$. Again, since $x$ is monotone increasing on $\gamma$, then $x$ evaluated on the omega--limit set $\omega(\gamma)$ must be constant, equal to the limit value of $x$ on $\gamma$. This implies that $x$ is constant on $\gamma_p$,
which gives
\begin{enumerate}
\item either  $x=0$
\item\label{itm:2} or $y=z=0$.
\end{enumerate}
Let us consider separately the two situations.
In the first case, using \eqref{eq:constr} in \eqref{eq:dy} we get
\begin{equation}\label{eq:5}
\frac{\mathrm dy}{\mathrm d\tau}=\left(3-\frac12 y^2\right)\left(\left(1-\frac\gamma 2\right)y-\alpha\right),
\end{equation}
and so $y(\tau)$ (which is bounded) tends to an equilibrium point of \eqref{eq:5}. Equilibria are given by $y=\pm\sqrt 6$ and $y=\tfrac{2\alpha}{\gamma-2}$; therefore, the only admissible case happens when $\tfrac{2\alpha}{\gamma-2}=0$, that is $\gamma=4/3$, since otherwise $y$ would tend to a nonzero value and, using \eqref{eq:dphi}, $\phi$ would be unbounded.
In the admissible case $\gamma=4/3$ it is easy to see that $\gamma_p$ reduces to the point $x=y=0$, $z=-\sqrt 3$ (we exclude the positive value by continuity). To prove that solutions of the system do not approach this equilibrium point generically, we use \eqref{eq:constr} in \eqref{eq:dx}--\eqref{eq:dy} and study the unconstrained system given by \eqref{eq:dphi}--\eqref{eq:dy}. It can be easily seen that eigenvalues of the equilibrium point $(\phi=\phi_0,x=0,y=0)$ are given by $0,-2$ and $1$,  so there exists an unstable manifold at the equilibrium with strictly positive dimension.

Then, it must be still analyzed case \eqref{itm:2} where $y=z=0$, which means $\phi=\phi_0$ constant and $x_0=-\sqrt{3/V(\phi_0)}$. This fact implies $V'(\phi_0)=0$ and it is better to use the constraint \eqref{fri1jm} to study the behavior of solutions of
\begin{align*}
&\dot\phi=v,\\
&\dot v=\sqrt{3\left(\rho+\frac12 v^2+V(\phi)\right)}\,v-V'(\phi)+\alpha\rho,\\
&\dot\rho=-\rho\left(\alpha v+\gamma\sqrt{3\left(\rho+\frac12 v^2+V(\phi)\right)}\right),
\end{align*}
in the unknown functions $(\phi,v,\rho)$ of $t$, near the equilibrium point $(\phi_0,0,0)$ (note that the system is regular at this point, since $V(\phi_0)$ is strictly positive). The eigenvalues are given by $\gamma\sqrt{3V(\phi_0)}$ and $\tfrac12\left(\sqrt{3V(\phi_0)}\pm\sqrt{3V(\phi_0)-4V''(\phi_0)}\right)$, which means that there exists at least one eigenvalues with positive real part and, again, there exists an unstable manifold at the equilibrium with strictly positive dimension. This completes the proof.
\end{proof}

Theorem \ref{thm:unbound} suggests that one needs to find a suitable
coordinate transformation, which maps a \emph{neighborhood} of $\phi=\infty$
onto a bounded set. Assumption \ref{as:phi} precisely performs this task, therefore from now on we will assume that the potential $V(\phi)$ belongs to the class $\mathfrak C_0$ (see Definition \ref{def:V0}).


Using the above hypothesis, we perform the change $s=f(\phi)$ and observe
that $\tfrac{\mathrm{d}s}{\mathrm{d}\tau}=-\tfrac{1}{h}f^{\prime}(\phi
)\dot{\phi}=-yf^{\prime}(\phi(s))$, where $\phi(s)$ is the inverse function of
$s(\phi)$ introduced in Assumption \ref{as:phi}. Observing that the above
relation, together with \eqref{eq:dz}, does not contain $x$, we can eliminate
$x$ from \eqref{eq:dy} using the constraint \eqref{eq:constr} and then we
end up at the following unconstrained system:
\begin{subequations}
\begin{align}
&  \frac{\mathrm{d}y}{\mathrm{d}\tau}=\left(  3-\frac{1}{2}y^{2}\right)
(y+\sqrt{6}u(\phi(s)))-z^{2}\left(  \frac{\gamma}{2}y+\alpha+\sqrt{6}%
u(\phi(s))\right)  ,\label{eq:dy2}\\
&  \frac{\mathrm{d}z}{\mathrm{d}\tau}-\frac{1}{2}z\left[  y^{2}-\alpha
y+\gamma(z^{2}-3)\right]  ,\label{eq:dz2}\\
&  \frac{\mathrm{d}s}{\mathrm{d}\tau}=-yf^{\prime}(\phi(s)).\label{eq:ds2}%
\end{align}
Because of assumption \ref{as:phi}, we can extend
the functions $u(\phi(s))$ and $f(\phi(s))$ up to $s=0$. Then, due to
constraint \eqref{eq:constr} and to the fact that $V(\phi)$ is positive as
$\phi\rightarrow\infty$, we will consider as admissible set for the system
\eqref{eq:dy2}--\eqref{eq:ds2}, the set
\end{subequations}
\[
\Omega_{+}=\{(y,z,s)\,:\,\tfrac{1}{2}y^{2}+z^{2}\leq3,\,s\in
\lbrack0,s_0]\}.
\]

The assumptions made on $V(\phi)$ allow us to
consider the same problem \eqref{eq:dy2}--\eqref{eq:dz2} on the set
$\Omega_{-}=\{(y,z,s)\,:\,\tfrac12 y^{2}+ z^{2}\le3,\,s\in[-s_0,0]\}$.
Let us observe that the set $\{s=0\}$ may be seen as a subset of both $\Omega_{+}$
and $\Omega_{-}$, but its meaning is  different between the two
situations, since it represents the flow at positive and negative infinity,
respectively. To distinguish between the two cases, we will refer to them as
the sets $\Sigma_{+}$ and $\Sigma_{-}$, respectively.

In this way we have reduced the study of our dynamical system in a compact
set. To study the global properties of the flow, however, we need to take into
account the behavior of the solution when the scalar field $\phi$ lies in the set
$[N,M]$. To this aim we first observe that, choosing initial data of
\eqref{eq:dphi}--\eqref{eq:dz} such that $x(0)=x_{0}<0$, then $x(\tau)^2<x_{0}^2, \forall\tau>0$. Therefore we need to perform a $\mathcal{C}^{2}$ junction of
the sets $\Omega_{+}$ and $\Omega_{-}$ described above, with a suitable compact subset
$\Omega_{0}$ of $\{(y,z,\phi)\,:\,\tfrac12 y^{2}+z^{2}\le 3+C x_{0}^{2}%
,\,\phi\in[N,M]\}$, where the constant $-C$ (that can be chosen to be nonnegative) bounds $V(\phi)$ from below, in
such a way that the flow for $t>0$ completely lies in the set $\Omega
=\Omega_{-}\cup\Omega_{0}\cup\Omega_{+}$, see Figure \ref{fig:Omega}.
\begin{figure}
\begin{center}
\psfull \epsfig{file=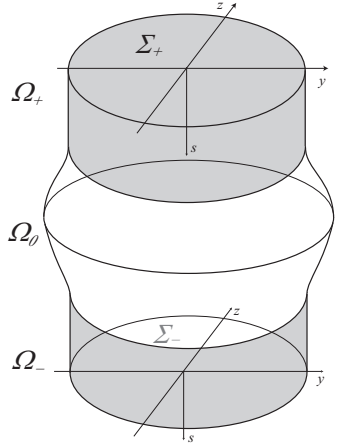, height=6cm} \caption{The set $\Omega$ is defined as the union of $\Omega_\pm$ with a suitable subset of
$\{(y,z,\phi)\,:\,\tfrac12 y^2+z^2\le 3+ C x_0^2,\,\phi\in[N,M]\}$. Note that coordinate chart on $\Omega_\pm$ and $\Omega_0$ are given by $(y,z,s)$ and $(y,z,\phi)$, respectively.}\label{fig:Omega}
\end{center}
\end{figure}

\begin{definition}\label{def:problequiv}
We say that a curve $\gamma:\mathbb{I}\rightarrow
\Omega$ is a solution of the problem $\mathcal{P}$ if it solves \eqref{eq:dy2}--\eqref{eq:ds2} when $\gamma$ lives in
$\Omega_{\pm}$ and it solves the system given by equation \eqref{eq:dphi},
\eqref{eq:dy} and \eqref{eq:dz} when $\gamma$ is
in $\Omega_{0}$, where $x$ is given by the constraint \eqref{eq:constr}.
\end{definition}

\begin{remark}
\label{rem:equiv} It is a simple consequence of the above construction, that
every solution of the initial problem \eqref{fri1jm}--\eqref{conssfjm} (with
$k=0$) such that $H(0)<x_{0}^{-1}$, corresponds uniquely to a curve $\gamma$
in $\Omega$ solution of $\mathcal{P}$. Of course, the converse may be false,
because the set of solutions of $\mathcal{P}$ may also include curves lying in
the set
\[
\uby=\Sigma_{+}\cup\{(y,z,s)\in\Omega_{+}\,:\,\tfrac12 y^{2}+z^{2}%
=3\}\cup\{(y,z,s)\in\Omega_{-}\,:\,\tfrac12 y^{2}+z^{2}=3\}\cup\Sigma_{-}.
\]
\end{remark}

The following lemma states a useful property of $\omega$--limit points of the flow.

\begin{lemma}
\label{thm:lemomega} Let $\gamma\in\Omega$ be a generic solution of problem
$\mathcal{P}$ and let $p\in\omega(\gamma)$ an $\omega$--limit point for
$\gamma$. If $p\in\Omega_{+}\cup\Omega_{-}$, then $p\in\uby$.
\end{lemma}

\begin{proof}
It must be shown that $\omega(\gamma)\cap(\Omega_+\cup\Omega_-)\subseteq\uby$. First, we observe that Theorem \ref{thm:unbound} states the existence of a point $q\in\omega(\gamma)$ such that $q\in\Sigma$. That means that $\lim_{t\to\sup\I}x(t)=0$, because the constraint \eqref{eq:constr} is satisfied on $\gamma$ and then on $\omega(\gamma)$. Therefore, if $p\in\omega(\gamma)\cap(\Omega_+\cup\Omega_-)$ is such that $s\ne 0$, then $V$ is finite and again from \eqref{eq:constr} we get $\tfrac12 y^2+z^2=3$, which implies $p\in\uby$.
\end{proof}

The above lemma states, in rough words, that a \emph{generic} solution of problem $\mathcal P$ approaches the set $\uby$. By the word ``generic'' we mean, in this context, a solution corresponding to a generic -- up to zero--measured set -- choice of initial data.
There can be solutions such that the scalar $\phi$ approaches, for instance, a local extremum of the potential $V(\phi)$, but these solutions correspond to  special choices of the initial data and are unstable with respect to this choices.

\subsection{The flow near $\uby$}\label{sec:equilibria}

In order to study late time behavior of collapsing solutions, we need information on the structure of $\omega(\gamma)$. Most of this information can be obtained studying the flow near $\uby$. In this section we will list equilibria
of the system \eqref{eq:dy2}--\eqref{eq:ds2} and their properties. It will turn out that the phase portrait significantly changes, especially at $\Sigma_\pm$, depending on the values of $\gamma$ and $u(0)$, the limit value of $u(\phi(s))$ as $s\to 0$. The different situations depending on these parameters is summarized in  Figure \ref{fig:zone}: each region corresponds to a particular phase portrait, sketched in Figure \ref{fig:pp}.

\begin{figure}
\begin{center}
\psfull \epsfig{file=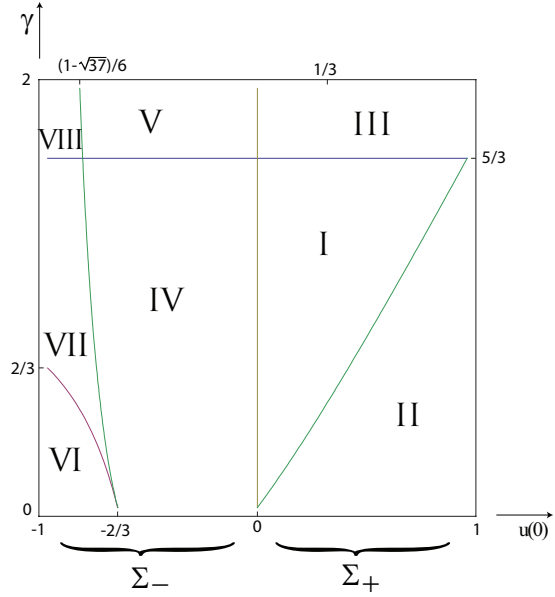, height=8cm} \caption{Summary of the different phase portrait of the flow near $\Sigma_\pm$, according to the values of $\gamma$ and $u(0^\pm)$.  The roman numbers correspond to situations depicted in Figure \ref{fig:pp}.}\label{fig:zone}
\end{center}
\end{figure}

In the following, we will denote points of $\Sigma_{\pm}$ by $(y,z,0^{\pm})$ .

\begin{figure}[htp]
  \begin{center}
    \subfigure[The only stable point is $\mathcal A_+$.]{\label{fig:1}\includegraphics[width=7cm]{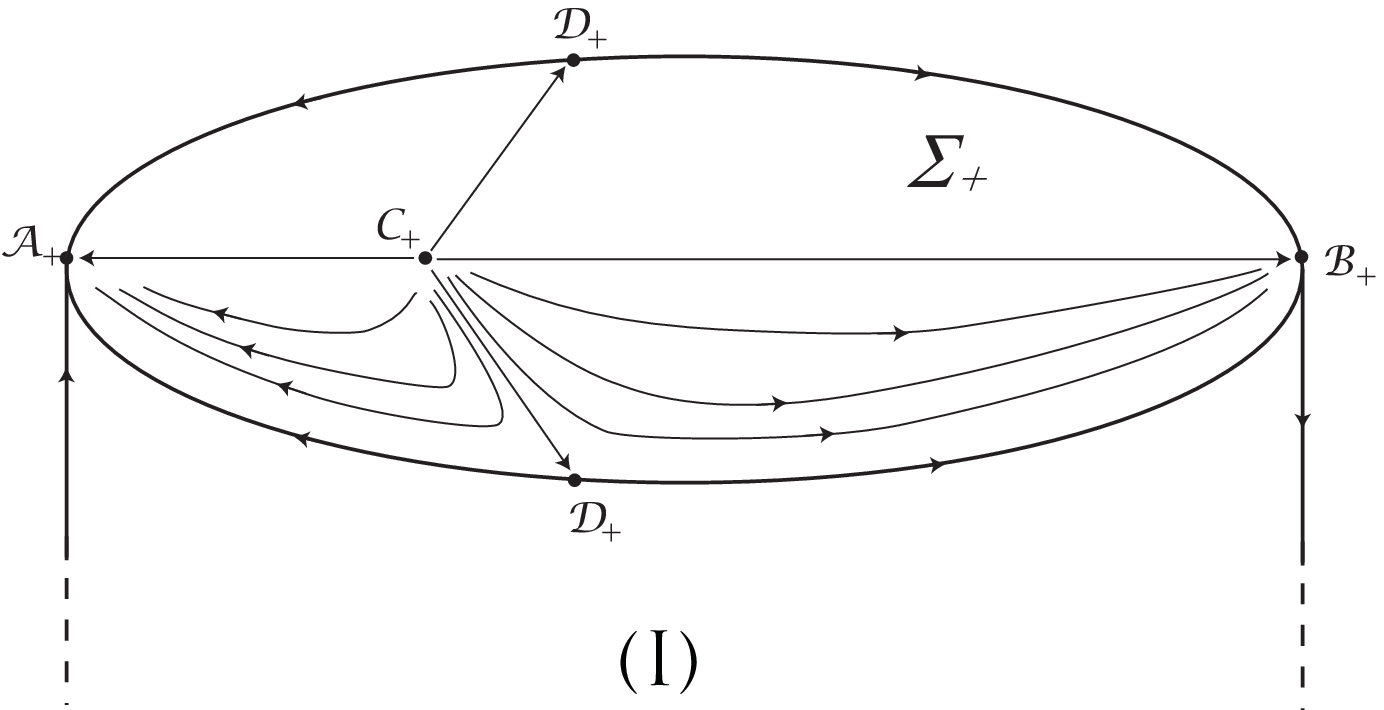}}
    \subfigure[The (unstable) equilibrium $\mathcal E_+$ is connected with $\mathcal C_+$ and $\mathcal D_+$ with two separatrices spiraling through each other.]{\label{fig:2}\includegraphics[width=7cm]{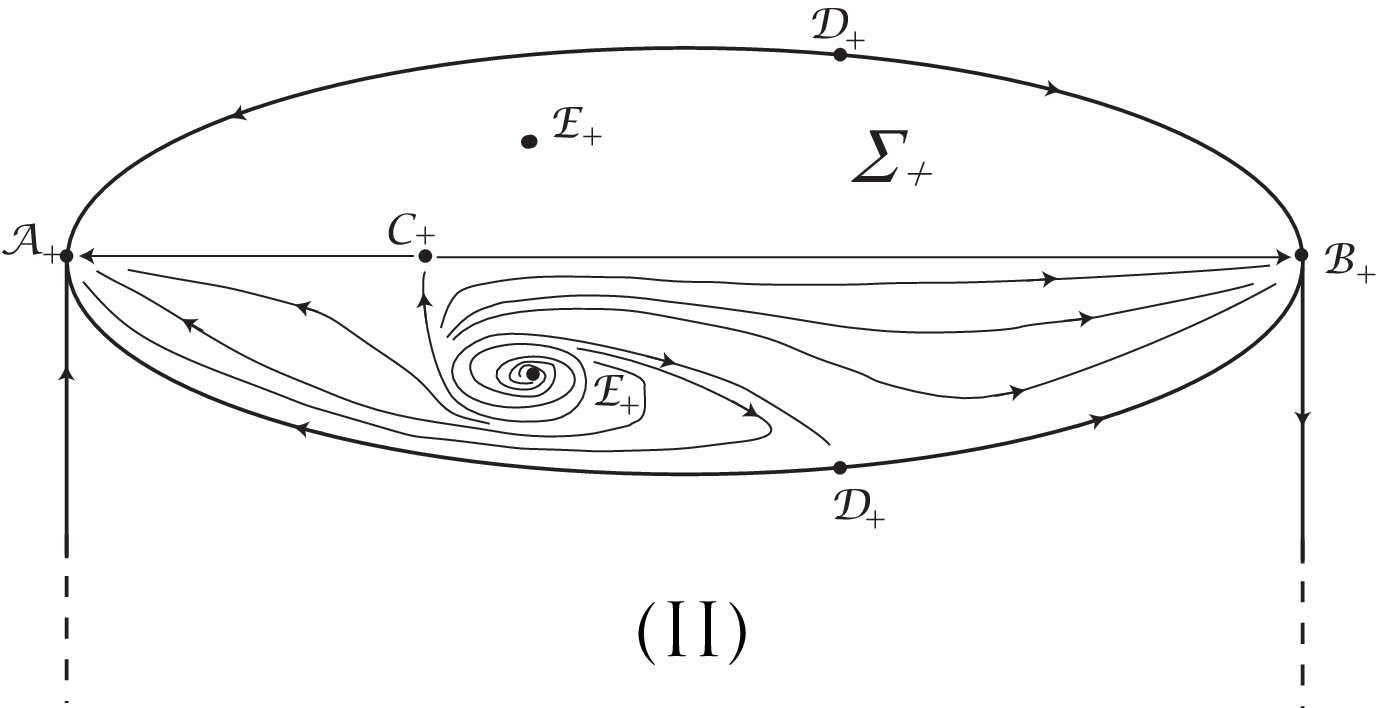}} \\
    \subfigure[No stable point exists: the flow is generically forced to move away from $\Sigma_+$.]{\label{fig:3}\includegraphics[width=7cm]{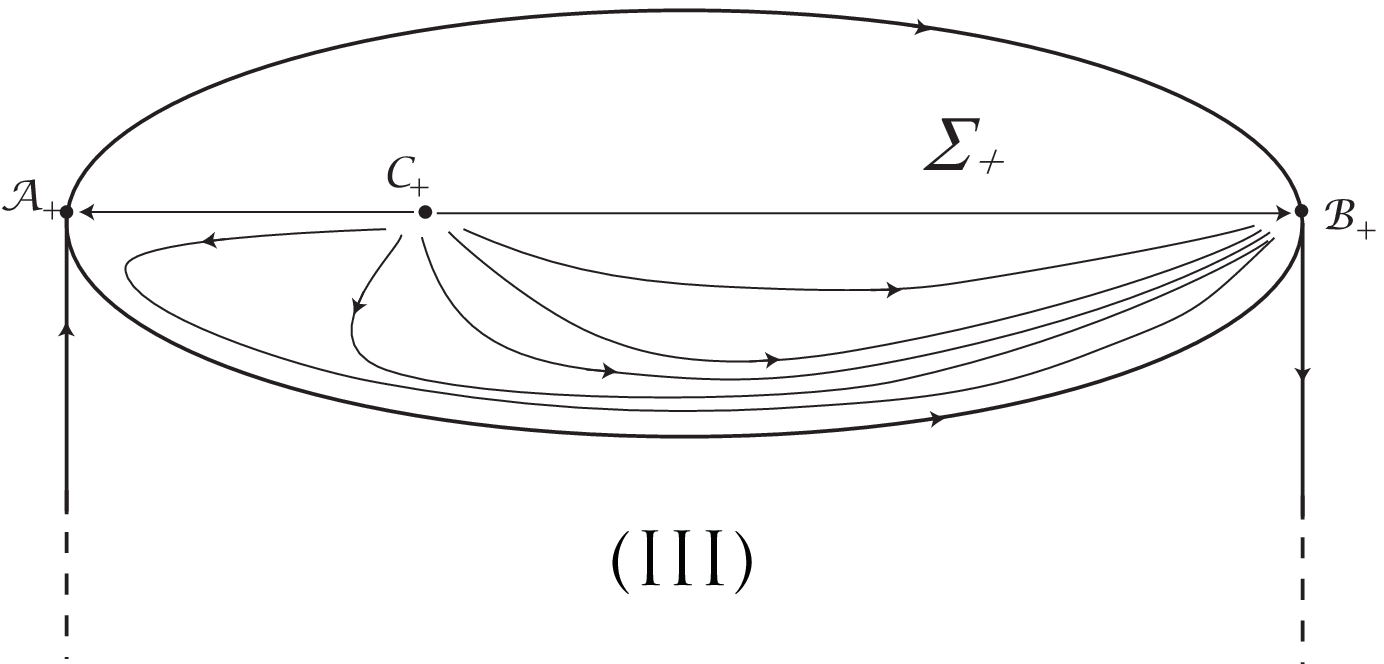}}
    \subfigure[The only stable equilibrium point is $\mathcal B_-$, but there are also generical solutions driving through $\mathcal A_-$ and moving away from $\Sigma_-$.]{\label{fig:4}\includegraphics[width=7cm]{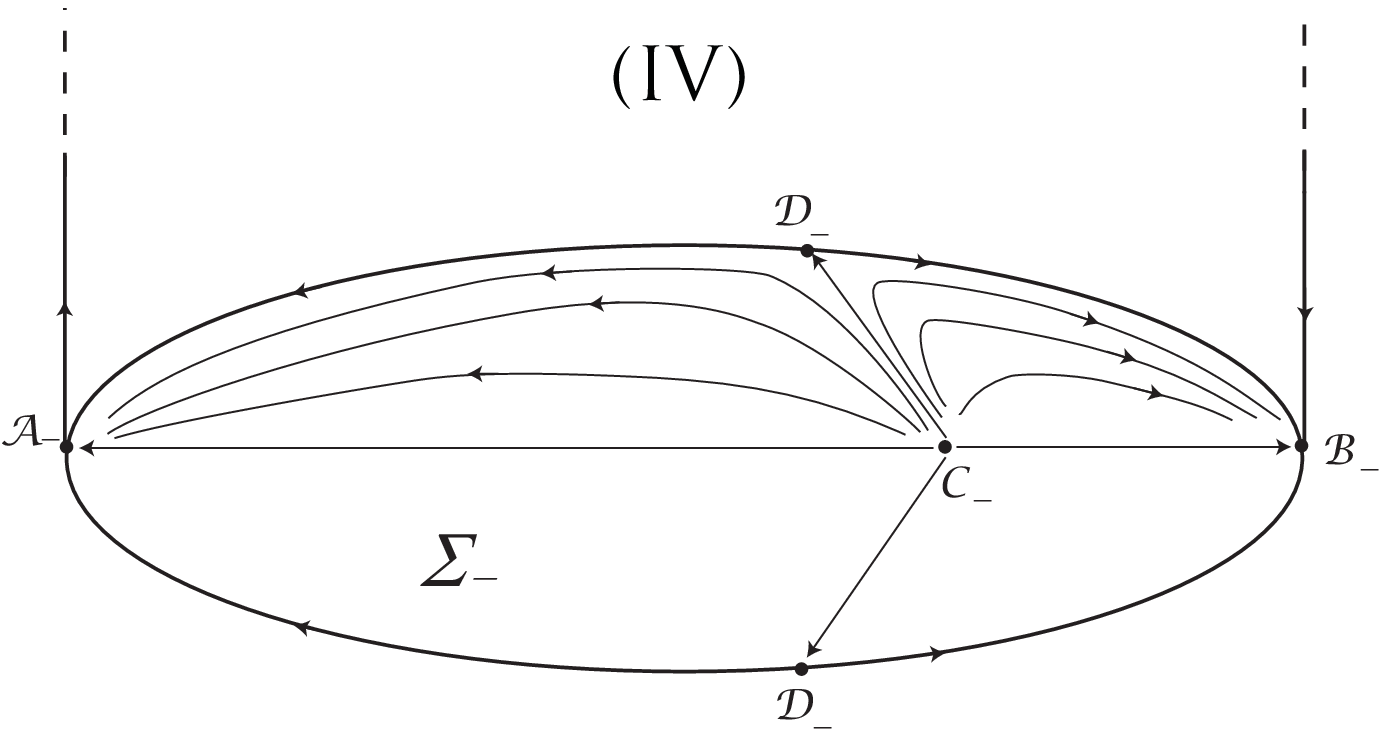}}\\
    \subfigure[Solutions are generically forced to approach $\mathcal B_-$.]{\label{fig:5}\includegraphics[width=7cm]{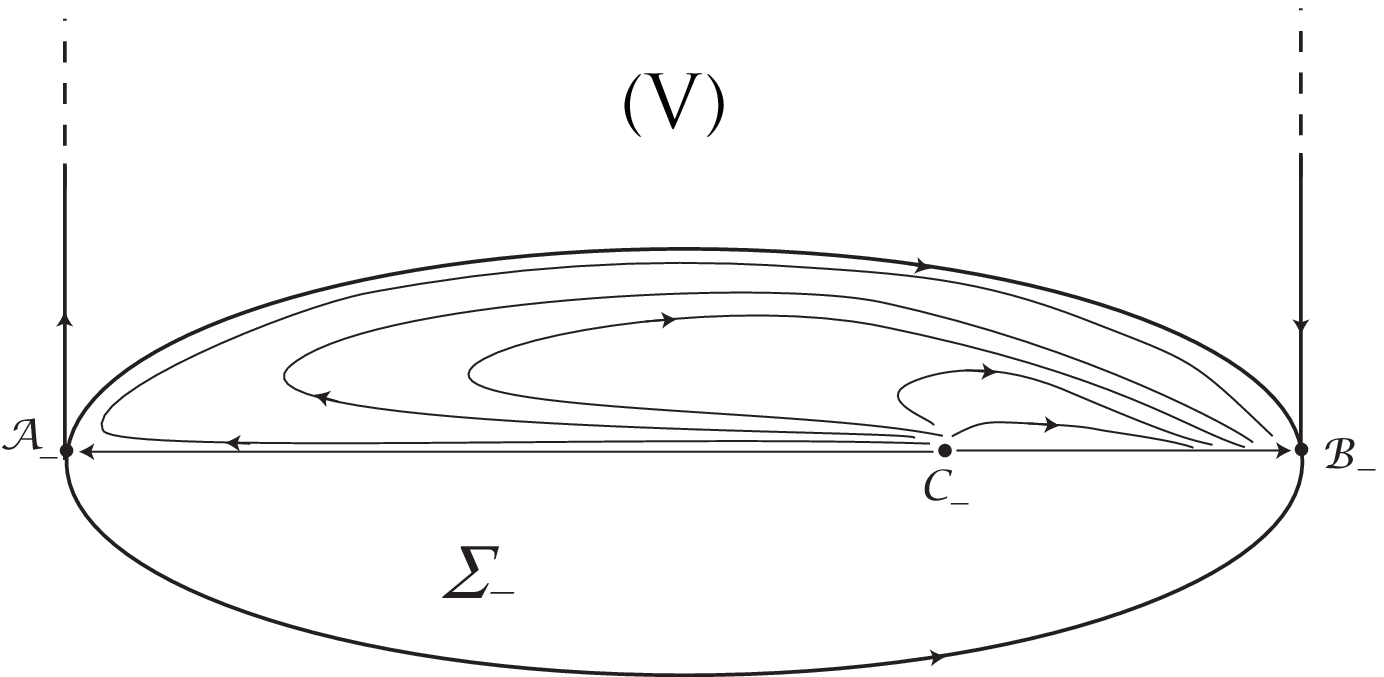}}
    \subfigure[Similar to Figure \ref{fig:4}, only direction of the separatrix between $\mathcal C_-$ and $\mathcal D_-$ changes.]{\label{fig:6}\includegraphics[width=7cm]{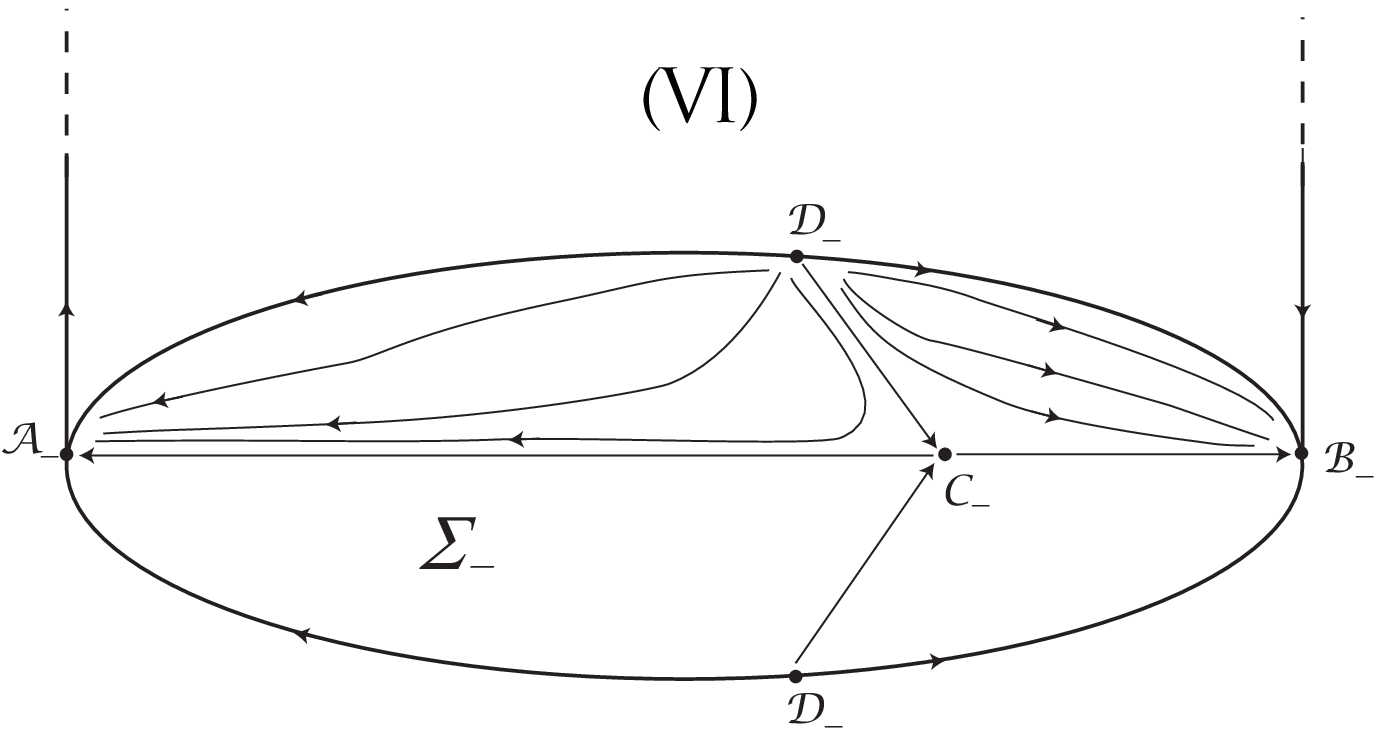}}\\
    \subfigure[The only stable equilibrium is $\mathcal B_-$ and $\mathcal E_-$ exhibits a similar behavior as $\mathcal E_+$ of Figure \ref{fig:2}.]{\label{fig:7}\includegraphics[width=7cm]{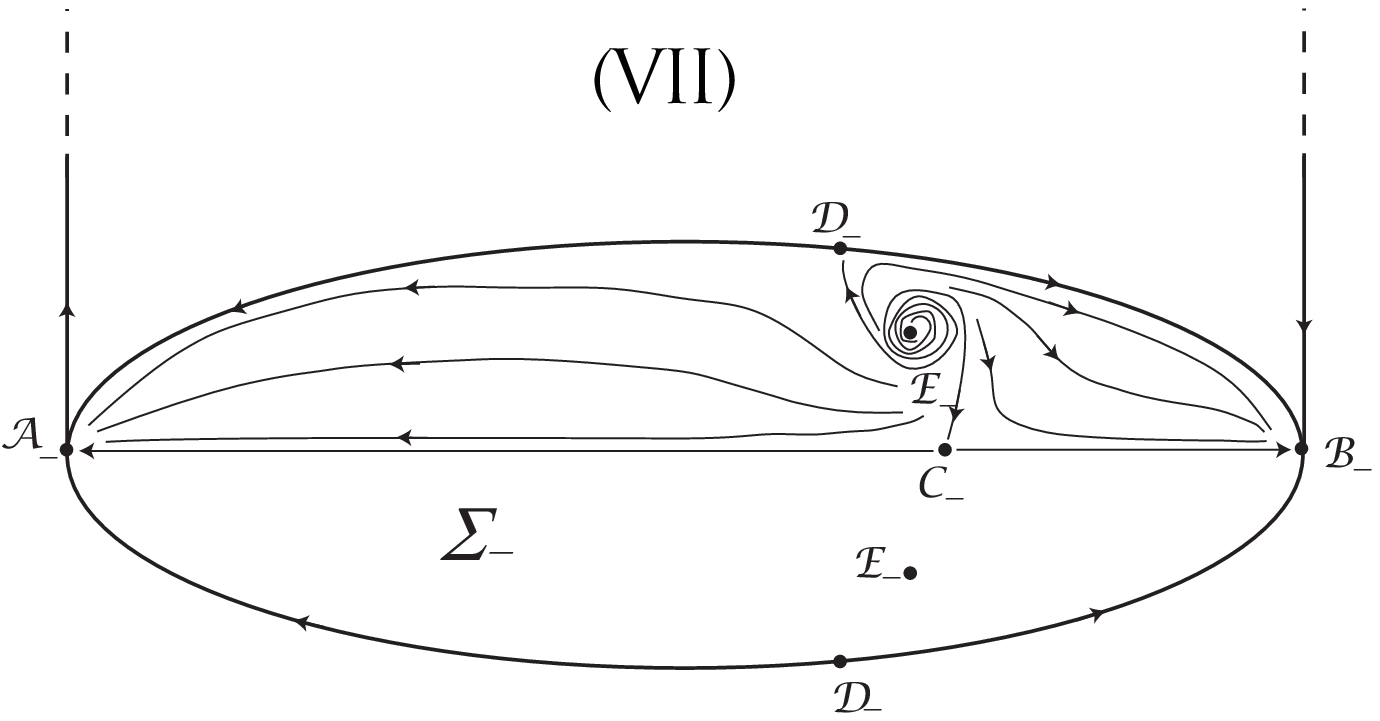}}
    \subfigure[Curves in $\Sigma_-$ are generically forced to move towards $\mathcal B_-$.]{\label{fig:8}\includegraphics[width=7cm]{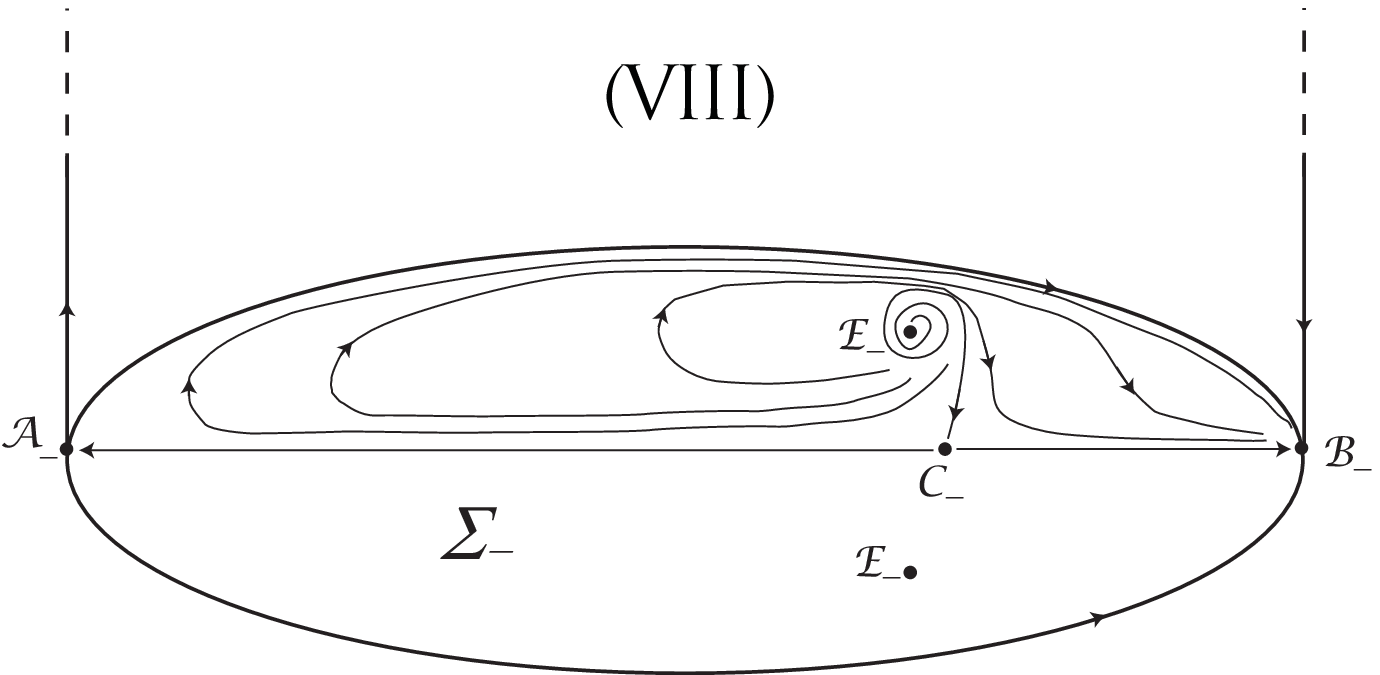}}
  \end{center}
  \caption{Phase portrait sketches for the different situations summarized in Figure \ref{fig:zone}.}
  \label{fig:pp}
\end{figure}

\begin{enumerate}
\item $\mathcal{A}_{\pm}(-\sqrt6, 0, 0^{\pm})$. The eigenvalues are
$\lambda_{1}=-5+3\gamma,\,\lambda_{2}= 6(u(0^{\pm})-1),\,\lambda_{3}=0$, with
corresponding eigenvectors $\mathbf{v_{1}}(0,1,0),\,\mathbf{v_{2}%
}(1,0,0),\,\mathbf{v_{3}}(0,0,1)$. To find the behavior of the flow in the direction
$\mathbf{v_{3}}$, note that \eqref{eq:ds2} implies $\frac{\mathrm{d}\phi
}{\mathrm{d}\tau}=\sqrt6$ on $\mathbf{v_{3}}$, from which we get $\phi
\approx\sqrt6\tau+\phi_{0}$. This means that $\mathcal{A}_{+}$ is attractive
along $\mathbf{v_{3}}$ and $\mathcal{A}_{-}$ is repulsive along the same
direction. Then we have the following three subcases, according to the value
of $\gamma$:

\begin{enumerate}
\item $\gamma>5/3$. The point $\mathcal{A}_{\pm}$ is unstable and
the only stable direction tangent to $\Sigma_{\pm}$ is given by $\mathbf{v_{2}}$.

\item $\gamma=5/3$. The central manifold is 2-dimensional and it easy to see
that $\tfrac{\mathrm{d}z}{\mathrm{d}\tau}=\sqrt6 z^{2} (1-u(s))>0$,  so the
point is repulsive along the direction $\mathbf{v_{1}}$.

\item $\gamma<5/3$. The point is attractive along directions in the plane
$\Sigma_{\pm}$ and so $\mathcal{A}_{+}$ is a stable equilibrium point
(whereas $\mathcal{A}_{-}$ is unstable because repulsive along $\mathbf{v_{3}%
}$).
\end{enumerate}

\item $\mathcal{B}_{\pm}(\sqrt6, 0, 0^{\pm})$. The eigenvalues are
$\lambda_{1}=-1,\,\lambda_{2}=-6(1+u(0^{\pm})),\,\lambda_{3}=0$, with
corresponding eigenvectors $\mathbf{v_{1}}(0,1,0),\,\mathbf{v_{2}%
}(1,0,0),\,\mathbf{v_{3}}(0,0,1)$. Arguing as before, $\phi\approx-\sqrt
6\tau+\phi_{0}$ along $\mathbf{v_{3}}$ and so $\mathcal{B}_{-}$ is a stable
equilibrium point, whereas $\mathcal{B}_{+}$ is unstable along the direction
$\mathbf{v_{3}}$.

\item $\mathcal{C}_{\pm}(-\sqrt6 u(0^{\pm}),0,0^{\pm})$. The eigenvalues are
given by $\lambda_{1}=3(1-u(0^{\pm})^{2}),\,\lambda_{2}=\tfrac12(3\gamma
-4u(0^{\pm})+3\gamma u(0^{\pm})-6u(0^{\pm})^{2}),\,\lambda_{3}=0$, with
corresponding eigenvectors $\mathbf{v_{1}}(0,1,0),\,\mathbf{v_{2}%
}(1,0,0),\,\mathbf{v_{3}}(0,0,1)$. The equilibrium point is unstable since
$\lambda_{1}>0$, but it is useful to observe stability along directions on
$\Sigma_{\pm}$ and so the sign of $\lambda_{2}$ must be investigated. Let us
distinguish the two following cases:

\begin{enumerate}
\item $\mathcal{C}_{+}$. In this case $\forall u=u(0^{+})\in]0,1[$, there
exists $\gamma_{\mathcal{C}}(u)\in]0,5/3]$ such that $\lambda_{2}\gtreqless0$
if $\gamma\gtreqless\gamma_{\mathcal{C}}(u)$. If $\gamma=0$, then
$\lambda_{2}>0,\,\forall\gamma\in]0,2]$.

\item $\mathcal{C}_{-}$. In this case $\forall u=u(0^{-})\in]-1,-{2/3}[$,
there exists $\gamma_{\mathcal{C}}(u)\in]0,2]$ such that $\lambda
_{2}\gtreqless0$ if $\gamma\gtreqless\gamma_{\mathcal{C}}(u)$. If $u(0^{-}%
)\in[-{2/3},0]$, then $\lambda_{2}>0$.
\end{enumerate}

In both cases, the expression of $\gamma_{\mathcal{C}}(u(0))$ is given by
$\gamma_{\mathcal{C}}(u)=\tfrac{2u(2+3u)}{3(1+u)}$.

\item $\mathcal{D}_{\pm}(\tfrac{2\alpha}{2-\gamma},+z_{\gamma},0_{\pm})$ and
$(\tfrac{2\alpha}{2-\gamma},-z_{\gamma},0_{\pm})$ are four equilibria, where
$z_{\gamma}$ is chosen such that the points lie on the set $\tfrac12
y^{2}+z^{2}=3$. These equilibria exist only if $\gamma\le5/3$. When
$\gamma=5/3$ we recover the case $\mathcal{A}_{\pm}$ discussed above. If
$\gamma<5/3$, the eigenvalues are given by $\lambda_{1}=0,\,\lambda_{2}%
=\tfrac{2(5-3\gamma)}{3(2-\gamma)}>0$ and $\lambda_{3}=\tfrac{2(-8+3\gamma-12
u(0^{\pm})+9\gamma u(0^{\pm}))}{3(2-\gamma)}$. The corresponding eigenvectors
are $\mathbf{v_{1}}(0,0,1),\,\mathbf{v_{2}}$ is the tangent vector to the
boundary of the set $\Sigma_{\pm}$ and $\mathbf{v_{3}}$ is transversal to the
same boundary. If $u=u(0)\in]-1,-2/3[$, there exists a $\gamma_{\mathcal{D}%
}(u)\in]0,2/3]$ such that $\lambda_{3}\gtreqqless0$ when $\gamma
\lesseqqgtr\gamma_{\mathcal{D}}(u)$. Otherwise, $\lambda_{3}<0$. The
expression for $\gamma_{\mathcal{D}}(u)$ is $\gamma_{\mathcal{D}}%
(u)=\tfrac{4(2+3u)}{3(1+3u)}$.

\item $\mathcal{E}_{\pm}(-\tfrac{3\gamma}{\alpha+\sqrt6 u(0^{\pm})},z_{0},0)$
and $(-\tfrac{3\gamma}{\alpha+\sqrt6 u(0^{\pm})},-z_{0},0)$ are other four
equilibria, lying on the interior of $\Sigma_{\pm}$. Let us distinguish the
cases $0^{+}$ and $0^{-}$:

\begin{enumerate}
\item On $\Sigma_{+}$, the points are admissible if $\gamma\le5/3$ (the limit
case $\gamma=5/3$ the two points coincides with $\mathcal{A}_{+}$) in
particular, $\forall u=u(0^{+})\in]0,1[$, the equilibria on $\Sigma_{+}$ are
admissible when $\gamma\le\gamma_{\mathcal{C}}(u)$ (in the limit case
$\gamma_{\mathcal{C}}(u)$ the two points coincide exactly with $\mathcal{C}%
_{+}$). Note that $y\le0$.

\item On $\Sigma_{-}$, $\forall u=u(0^{-})\in]-1,-{2/3}[$, the points are
admissible if $\gamma_{\mathcal{D}}(u)\le\gamma\le\gamma_{\mathcal{C}}(u)$,
the limit cases coinciding with $\mathcal{D}_{-}$ and $\mathcal{C}_{-}$.
Notice that $y\ge0$.
\end{enumerate}

\item \label{itm:g43} Finally, only in the case $\gamma=4/3$, there exists a
set of equilibria given by the points $(0,\pm\sqrt3, s)$, $\forall s$. Note
that, at level $s=0$, that gives exactly the points $\mathcal{D}_{\pm}$. The
eigenvalues are given by $\lambda_{1}=-4,\,\lambda_{2}=1,\,\lambda_{3}=0$,
with eigenvectors  $\mathbf{v_{1}}$, corresponding to a transversal direction
to the boundary of the level set $s=\mathrm{const.}$, $\mathbf{v_{2}}$ the
tangent direction to the same boundary, and $\mathbf{v_{3}}(0,0,1)$.
\end{enumerate}

\subsection{Singularity formation}\label{sec:main}

The analysis of the flow near the set $\uby$ allows to complete the study of late time behavior of the solution.
\begin{theorem}\label{thm:main}
Let $V\in\mathfrak C_0$ and  $x_0<0$. For almost every choice of initial data for the system \eqref{fri1jm}--\eqref{conssfjm} with $k=0$ and $H(0)\in[x_0^{-1},0[$, $\exists t_s>0$ such that the solution admits a right maximal extension to the set $[0,t_s[$. The scalar field $\phi$ diverges in modulus as $t\to t_s^{-}$, together with its velocity $\dot\phi$.
\end{theorem}

\begin{proof}
First, we cast the system into the compact framework of problem $\mathcal P$ (see Definition \ref{def:problequiv}).
Using the analysis performed in last Section, we know that there are three possible cases in $\Sigma_{+}$ and five possible cases in
$\Sigma_{-}$, for the flow of the system at $x=0$. By inspection
of all possible cases, it is a lengthy but straightforward task to conclude that there cannot exist cycles
contained in the set $\uby$ and this means that each solution of problem
$\mathcal{P}$ can only approach the stable equilibria of the system, which are  the points $\mathcal{A}_{+}$ and
$\mathcal{B}_{-}$ lying on the ``unphysical'' boundary $\uby$. Let us observe also that $\mathcal A_+$ is admissible only when $\gamma<5/3$.

Translated into the original problem formalism with unknown functions depending on comoving time $t$, the aforesaid means that there are two possible generical endstates for collapsing solutions:

\begin{enumerate}
\item \label{itm:plus} either $\phi\to+\infty$, $\dot\phi\to+\infty$,
$\tfrac{\dot\phi}{H}\to-\sqrt6$ and $\tfrac{\sqrt\rho}{H}\to0$,

\item \label{itm:minus} or $\phi\to-\infty$, $\dot\phi\to-\infty$,
$\tfrac{\dot\phi}{H}\to\sqrt6$ and $\tfrac{\sqrt\rho}{H}\to0$.
\end{enumerate}

In both cases $H$ diverges to $-\infty$ and it is easy to show that this
happens as the comoving time approaches a finite value $t_{s}>0$, resulting in
singularity formation. Indeed, from \eqref{fri2jm} (recall $k=0$) and
\eqref{eq:change}, we have
\begin{equation}\label{eq:hdiv}
\dot H=-\frac{H^{2}}2\left(  y^{2}+\gamma z^{2}\right)  .
\end{equation}
Since the quantity in round brackets approaches the positive value 6, there
exists some $t_{0}>0$ such that $\dot H(t)<-2 H(t)^{2}$ and therefore, by standard comparison argument in ODE theory, the
function $H(t)$ is not extendable beyond some finite value $t_{s}$.
\end{proof}

\begin{remark}\label{eq:comparison}
As seen before, the singularity forms because the scalar field energy
$\frac12\dot\phi(t)^{2}+V(\phi(t))$ diverges. It is interesting to investigate
the late time behavior of the fluid energy $\rho(t)$. This can be easily
inferred by inspection of \eqref{conssfjm}, recalling $k=0$ and
\eqref{eq:change}:
\[
\frac{\dot\rho}{\rho}=-\dot\phi\left(  \frac{3\gamma}y+\alpha\right)  .
\]
In case \eqref{itm:minus}, the quantity into round brackets approaches the
value $4/\sqrt6$ and therefore $\rho(t)\approx e^{-4\phi/\sqrt6}$ as $\phi
\to-\infty$,  so the energy density diverges. Instead, in case
\eqref{itm:plus} (which can happen only if $\gamma<5/3$) $\rho(t)\approx
e^{\phi(6\gamma-4)/\sqrt6}$, which implies that the energy density of the
fluid diverges only if $\gamma>2/3$ and is bounded otherwise. However, since
$\tfrac{\rho}{H^{2}}$ tends to 0, the energy density of the fluid produces a
lower order curvature divergence with respect to the energy of the scalar
field, and the latter dominates in the late time of the collapse.
\end{remark}

\section{Gravitational collapse models}\label{sec:star}

The model studied attracts a great deal of attention in cosmology, since it can be seen as the study of \emph{backwards--in--time} singularity
for a FRW spacetime in an expanding phase --  in other words, reversing time direction we can study the initial singularity of expanding cosmologies. Nevertheless, we can also consider these singular solutions as matter interior, performing a suitable match with an exterior spacetime, thereby obtaining models of collapsing objects, whose endstate can be investigated.

The natural choice for the exterior is
the so--called \emph{generalized Vaidya solution}, which is
essentially the spacetime generated by a radiating object (see \cite{ww} for a detailed description),
\begin{equation}\label{eq:Va}
\text ds_{\mathrm{ext}}^2=-\left(1-\frac{2M(U,Y)}Y\right)\,\mathrm
dU^2-2\,\mathrm
dY\,\mathrm dU + Y^2\,\mathrm d\Omega^2,
\end{equation}
where $M$ is an arbitrary (positive) function.
The matching is performed along a hypersurface $\Sigma=\{r=r_b\}$. The Israel junction conditions at the
matching hypersurface with a FRW interior have been shown in \cite{giambo1,collapse} and amounts to require that Misner--Sharp mass $M$ is continuous and
$\tfrac{\partial M}{\partial U}=0$ on the junction hypersurface. Note that the equation of motion for the scalar field
remains smooth on the matching hypersurface.

The endstate of the collapse of these homogeneous perfect fluid
stars is analyzed in the following theorem (compare with \cite[Theorem 4.1]{collapse}).

\begin{teo}\label{thm:endstate}
A homogeneous perfect fluid star generically collapses to a black hole.
\end{teo}
\begin{proof}

The equation of the apparent horizon for
the flat FRW metric  is given by $r^2\dot a(t)^2=1$. If $\dot a$
is bounded, the junction surface $r=r_b$ can be chosen sufficiently
small such that  $\left(1-\tfrac{2M}R\right)$
is bounded away from zero near the singularity. As a consequence,
 the exterior portion of the spacetime admits
null radial geodesics that can be extended in the past up to the singularity which is therefore naked; otherwise, if $\dot a^2$ is
unbounded, the trapped region forms and the collapse is covered into a
black hole \cite{giambo1}.

Now, using \eqref{fri1jm}--\eqref{fri2jm} together with \eqref{eq:change}, one has
\begin{equation}\label{eq:ddota}
\frac{\ddot a(t)}{\dot a(t)}=\frac{\dot a(t)}{a(t)}\left[1-\frac12(y^2+\gamma z^2)\right],
\end{equation}
and since, as a consequence from Theorem \ref{thm:main}, the quantity into square brackets converges generically to the value $-2$,
we can integrate \eqref{eq:ddota} with respect to $t$ and using the fact that $a(t)\to 0$ we conclude that $\dot a(t)$ generically diverges to $-\infty$ and hence the singularity is hidden inside a black hole for almost each choice of the initial data.
\end{proof}

\begin{remark}
For the models under exam, the formation of naked singularities is forbidden for almost every choice of the initial data. Nevertheless, there may be special situations where the apparent horizon does not form and the singularity is hence naked -- simply, the theorem excludes that this  happens generically.

For instance, with reference to the discussion of the equilibria given in Subsection \ref{sec:equilibria}, we know that there exists non generical solutions tending to the equilibrium $\mathcal C_+$ (Figure \ref{fig:2}), such that $y^2+\gamma z^2\to 6 u(0^+)$. Using \eqref{eq:hdiv} and \eqref{eq:ddota} respectively we have that, if $0<u(0^+)<\sqrt 3/3$, the singularity forms in a finite amount of comoving time and  is naked because $\dot a$ remains bounded. Since $z\to 0$ goes to zero, this case basically corresponds to the special case in the situation where only the scalar field appears, and not the perfect fluid -- indeed, compare with \cite[eqns.(5.3)--(5.5)]{collapse}).
There may be also special solutions future asymptotic to $\mathcal D_+$, but in this case, anyway, the quantity $y^2+\gamma z^2$ is seen to converge to the value $\tfrac{16-6\gamma}{6-3\gamma}$. Then \eqref{eq:hdiv} ensures singularity formation in a finite time, but the collapse ends into a black hole, as can be seen using again \eqref{eq:ddota}.

In summary, we can conclude that the behavior here exhibited is -- for the same class of potentials -- quite similar to the scalar field without matter studied in \cite{collapse}. As well known, Oppenheimer--Snider model of homogeneous dust represents a classical case of covered singularity and this feature can be removed introducing higher order terms in the mass functions that destroy homogeneity. We recall that dust spacetimes can be regularly matched with Schwarzschild solution and, in this context, \eqref{eq:Va} can be seen as a generalization of Schwarzschild. Therefore, it can be said  the models studied in the present paper can be considered a scalar field generalization of the classical Oppenheimer--Snider solution, and the addition of a barotropic perfect fluid to the matter content has not affected the general picture, that basically remains the same of the classical massless case.
\end{remark}

\thanks{\textbf{Acknowledgment.} The author wishes to thank John Miritzis for stimulating discussions and valuable suggestions.}

\end{document}